\documentclass{ifacconf}
\usepackage{amsmath,amsfonts,amssymb}

\usepackage{graphicx,subfigure} % include this line if your document contains figures
\usepackage{natbib}             % required for bibliography
\usepackage{enumerate}
\usepackage[latin1]{inputenc}

\newlength{\dhatheight}

\newtheorem{definition}{Definition}
\newtheorem{remark}{Remark}

\newtheorem{theorem}{Theorem}
\newtheorem{lemma}{Lemma}

\newenvironment{proof}{
\noindent
    {\bf Proof:}
  }
  {
\vspace*{0pt}\hfill$\blacksquare$
  }

\newcommand\blfootnote[1]{%
  \begingroup
  \renewcommand\thefootnote{}\footnote{#1}%
  \addtocounter{footnote}{-1}%
  \endgroup
}

\begin{document}
\begin{frontmatter}

\title{Improving Linear State-Space Models with Additional Iterations}

%{\small Comment on the title: ``Full Parametrization'' refers to local LTI models and implies no global LPV structure assumption }

%\thanks[footnoteinfo]{This work has been partly supported by the ITEA2 MODRIO project.}
\author[Second+1]{Suat Gumussoy}
\author[Second]{Ahmet Arda Ozdemir}
\author[First+1]{Tomas McKelvey}
\author[First]{Lennart Ljung} 
\author[First+2]{Mladen Gibanica}
\author[Third]{Rajiv Singh} 

 %\address[First-1]{The Mathworks, Inc, Natick, MA, USA (email: suat.gumussoy@mathworks.com)}

\address[First]{Div. of Automatic Control, Link\"oping University,
  Sweden  (lennart.ljung@liu.se)}
\address[First+1]{Chalmers University of Technology (tomas.mckelvey@chalmers.se)}
\address[First+2]{Chalmers University of Technology \& Volvo Car
  Corporation, G\"oteborg, Sweden (mladen.gibanica@chalmers.se)}
 \address[Second+1]{MathWorks, Natick, MA, USA (suat.gumussoy@mathworks.com)}
\address[Second]{MathWorks, Natick, MA, USA (arda.ozdemir@mathworks.com)}
\address[Third]{MathWorks, Natick, MA, USA (rajiv.singh@mathworks.com)}
% \address[third]{Div. of Automatic Control, LinkÃ¶ping University, Sweden  (email: tschen@isy.liu.se)}

\begin{abstract}                % Abstract of not more than 250 words.
  An estimated state-space model can possibly be improved by further iterations
  with estimation data. This contribution specifically studies if models
  obtained by subspace estimation can be improved by subsequent re-estimation of
  the $B$, $C$, and $D$ matrices (which involves linear estimation problems).
  Several tests are performed, which shows that it is generally advisable to do
  such further re-estimation steps using the maximum likelihood
  criterion. Stated more succinctly in terms of
  \textsc{MATLAB}$^\text{\textregistered}$ functions, \texttt{ssest} generally
  outperforms \texttt{n4sid}.
\end{abstract}


\end{frontmatter}
%===============================================================================

\blfootnote{ MATLAB$^{\text{\textregistered}}$ and
  Simulink$^{\text{\textregistered}}$ are registered trademarks of The
  MathWorks, Inc. See mathworks.com/trademarks for a list of additional
  trademarks. Other product or brand names may be trademarks or registered
  trademarks of their respective holders.

  \hspace{.08in} Data for Section~\ref{sec:xc90}, Fig.~\ref{fig:xc90}, and
  Fig.~\ref{fig:xc90_2} presented in this paper is provided courtesy of Volvo
  Car Corporation. Presentation of the plot and the study results do not grant
  authorization to extracting and reusing the data for any purpose. Extraction
  and reuse of the data require authorization from Volvo Car Corporation.}

% \cite{Bar-ShalomLiKirubarajan2001}
\section{Introduction}
Linear state-space models are perhaps the most common model structure used in
system identification. They can be estimated in several different ways. Among
the most common techniques are so called subspace methods, such as MOESP,
\cite{VerhaegenD:92}, and N4SID, \cite{OverscheeDM:94}.  These were also
developed to work with frequency domain data,
\cite{McKelveyAkcayLjung:96}. Whichever way the model was obtained, an
interesting the question is if it can be improved in some way by further
polishing using observed data.

In this contribution we investigate how further iterations on the $B$, $C$
and $D$ matrices can possibly improve the model quality and also
whether it is worthwhile to reestimate the $A$ matrix.

We argue that such further estimation work is well motivated. The model
properties are improved most of the time (but not always), and in some cases,
the improvements are significant.
\section{The State-Space Model}
A linear state-space model in output error form (no noise model) is
given by
\begin{subequations}
\label{eq:1}
\begin{align}
  %\label{eq:1}
  x(t+1)&=Ax(t)+Bu(t)\\
y(t)&=Cx(t)+Du(t)+e(t)
\end{align}
\end{subequations}
By assuming Gaussian noise distribuion for $e$, the maximum likelihood estimate
(MLE) of the matrices is given by
\begin{align*}
 % \label{eq:2}
  \min_{A,B,C,D}&\sum_{t=1}^N \|y(t)-C\hat x(t)-Du(t)\|^2\\
&\hat x(t+1)=A\hat x(t) +Bu(t)
\end{align*}
which can be rewritten as
\begin{subequations}
\label{eq:ml}
\begin{align}
  \label{eq:3}
 \min_{A,B,C,D} &V(y,u,A,B,C,D)\\ 
\label{eq:3b}
V(y,u,A,B,C,D)& =\sum_{t=1}^N\|y(t) -\hat y(t|A,B,C,D)\|^2\\ \label{eq:3c}
\hat y(t|A,B,C,D)&=Du(t)+C\sum_ {k=1}^t A^{t-k}Bu(k)
\end{align}
\end{subequations}
We note that for fixed $A,C$ the prediction $\hat y$ is linear in $B,
D$ and for fixed $A,B$ the prediction is linear in $C,D$.
\begin{remark}
\label{rem:th}
The matrices $A,B,C,D$ may be parameterized by some paremeter vector
$\theta$. The estimation and minimization is carried out with respect to
$\theta$. It will be assumed that any such parameterization is linear in
$\theta$.
\end{remark}

\section{Refining the Estimate}
With a given state-space model (\ref{eq:1}) and access to data $[y,u]$ from the
system (either data that was used to estimate (\ref {eq:1}) or a fresh dataset),
it is natural to ask if the maximum likelihood (ML) method in Eq. (\ref{eq:ml})
can be used to refine the estimate.

$\hat y(t|A,B,C,D)$ is linear in $B,D$ for fixed $A,C$, and therefore
reestimating $B,D$ is a simple linear regression problem. Once this is done, one
could fix $A,B$ and reestimate $C,D$ which again is a linear regression
problem. Repeating that gives some simple $B,C,D$ iterations, which is a way of
minimizing (\ref{eq:ml}a) for a fixed $A$, which is a bilinear problem. See, for
example, Eq. (10.68) in \cite{Ljung:99}. Simple experiments show that the
$B,C,D$ iterations typically lead to the same (possibly local) minimum of
Eq. (\ref{eq:ml}a) as applying the Gauss-Newton method to minimize
Eq. (\ref{eq:ml}a) for fixed $A$.

The aforementioned $B, C, D$ iterations is a \emph{block coordinate descent}
approach. It is faster than classical coordinate descent since in two steps all
coefficients in $B,C,D$ are updated and at each iteration more accurate values
are used from previous iteration: updated $C$ in $B,D$ iteration and updated $B$
in $C,D$ iteration. The same approach is called \emph{separable least squares}
for a slightly different cost function appearing in identification of
Hammerstein systems, \cite{bai:2004convergence}. The convergence of block
coordinate descent is guaranteed when the objective function has a unique
minimum in each coordinate block, see pages 261-262 in
\cite{luenberger:2008}. This condition corresponds to the least squares problem
in Eq. \ref{eq:ml} not being rank deficient during each $B,D$ and $C,D$
iteration. This is frequently satisfied in practice per our experiments.

After the refinement of $B,C,D$ is done, one can indeed fix these parameters and
reestimate $A$ using Eq. (\ref{eq:ml}a) -- or alternatively minimize the ML
criterion with respect to all matrices simultaneously.
% \section{Code for Fixed $A$}
% \begin{verbatim}
% mn = n4sid(z,n,'DisturbanceModel','None');
% mn.Structure.A.Free=zeros(n,n);
% m = ssest(z,mn);
% \end{verbatim}
%\end{document}

\section{SISO, MISO, SIMO, and MIMO}

A typical implementation of subspace methods extract an estimate of the $A,C$
matrices first, followed by a linear regression for the $B$ and $D$ matrices
with fixed $A,C$. The result is not necessarily a (local) optimum in terms of
$A,C$ for MIMO systems per the cost function Eq. \eqref{eq:3}, and refinements
are possible. 

However, there are two special cases:

\begin{itemize}
\item SISO and MISO systems: The solution of the linear regression for $B,D$ is
  also a (local) optimum in terms of $C$. In other words, unless $A$ is refined,
  further iterations for $B$, $C$, $D$ matrices will not reduce the cost function.

\item SIMO systems: The solution of the linear regression for $B,D$ is not a
  local optimum in terms of $C$. However, solving one additional linear
  regression problem for $C, D$ matrices with fixed $A,B$ corresponds to a
  (local) optimum with respect to the $B$ matrix. After this additional linear
  regression, further iterations on $B,C,D$ matrices are not needed unless $A$
  is changed.
\end{itemize}

To clarify the statements above we have the following results:

\begin{definition}
  Two state-space realizations are \emph{input-output equivalent} if
%  their transfer functions are equal or equivalently 
their impulse
  responses (Markov parameters) are equal.  State-space realization
  $(A,B,C,D)$ is input-output equivalent to $(A_0,B_0,C_0,D_0)$ if
  $D_0=D$ and $CA^{i}B = C_0A_0^{i}B_0$ for all non-negative integers
  $i$.%  or equivalently 
    % \begin{equation}
    %   \label{eq:io_equiv}
    %   D + C(zI - A)^{-1} B = D_0 + C_0(zI-A)^{-1} B_0 
    % \end{equation}
\end{definition}
  \begin{lemma} \label{lem:1}
    Let $(A,B,C,D)$ be an order $n$ state-space realization of
    a MISO linear system.  Let $C_0$ be any row vector of the
    same size as $C$ such that $(A,C_0)$ is an observable pair. Then
    there exists matrices $B_0$ and $D_0$ such that the realization
    $(A,B_0,C_0,D_0)$ is input-output equivalent with $(A,B,C,D)$.
    % , i.e
    % \begin{equation}
    %   \label{eq:io_equiv}
    %   D + C(zI - A)^{-1} B = D_0 + C_0(zI-A)^{-1} B_0 
    % \end{equation}
  \end{lemma}
\begin{proof}
 Trivially we have the choice $D_0=D$ and we will show that there exists a
  $B_0$ such that the rest of the impulse response matrices coincide. % ,
  % i.e.
  %    \begin{equation}
  %      \label{eq:5}
  %      CA^iB = C_0 A^iB_0,  \quad \forall i \geq 0.
  %    \end{equation}
     Pick any matrix $P$ such that $CP=C_0$.
     Consider the construction $T=\sum_{i=0}^{n-1}  t_i A^{i}$ where
     the scalars $t_i$ are selected such that 
     \begin{equation}
 C = C_0T = \sum_{i=0}^{n-1} t_i C_0 A^{i}. 
\label{eq:16}
\end{equation}
Since the pair $(A,C_0)$ is observable the corresponding observability matrix
has full rank which imply that the row vectors $C_0A^{i}$ to the right in
\eqref{eq:16} are all linearly independent. Hence a solution exists. Note that
by construction $TA=AT$.  Finally we have the identity
 \begin{equation}
    \label{eq:17}
    CA^{i}B = CPTA^{i}B = C_0A^{i} TB
 \end{equation}
which directly shows that $B_0 = TB$.
\end{proof}

\begin{theorem} \label{thm:AC}
  Consider the case when the $A$ matrix is fixed and we define
  $(A,B_*,C_*,D_*)$ to be a state-space realization which
  minimizes~\eqref{eq:3b}, w.r.t. $A,B,C$ with $A$ fixed. Let $C_0$ be any row vector
  such that $(A,C_0)$ is observable. Define $B_0$ and $D_0$ to be the
  minimizers to ~\eqref{eq:3b} w.r.t.\ $B$ and $D$, when $A$ and $C_0$
  are fixed.  Then  $ V(y,u,A,B_*,C_*,D_*) = V(y,u,A,B_0,C_0,D_0)$.
\end{theorem}
\begin{proof}
 By Lemma~\ref{lem:1} we can conclude
  that for the realization $(A,B_*,C_*,D_*)$ there
  exists an input-output equivalent realization $(A,B',C_0,D')$
  which imply \\
  $ V(y,u,A,B_*,C_*,D_*) = V(y,u,A,B',C_0,D') $ since the loss function only depend
  on the input output properties and not on the specific
  realization. Since the minimization of $V$ w.r.t. $B$ and $D$  
  only has a subset of the free parameters compared to the
  minimization w.r.t. $B,C,D$  it follows that $B'$ and $D'$
  are the minimizer to $V$, i.e.\ $B_0=B'$ and $D_0= D'$.
\end{proof}

The results above can also be formulated and proved for SIMO systems. 

Consider the following result.
\begin{theorem} \label{thm:commute}
  Given a matrix $A\in\mathbb{R}^{n \times n}$  and assume there exists a vector $v$ such that $(A,v)$ is observable.
%Assume the minimal and characteristic polynomials to
%  $A\in\Real^{n\times n}$  coincide. 
Then for a matrix $B$ the following two statements are equivalent
\begin{enumerate}[(i)]
\item $AB = BA$
\item There exists $b_i\in\mathbb{R}$, $i=0,\ldots,n-1$ such that  $B=\sum_{i=0}^{n-1}b_iA^{i}$
\end{enumerate}
\end{theorem}

\begin{proof}
If (ii) is true then (i) follows immediately. 
Assume (i) is true. 
Since $(A,v)$ is observable, the set  $ \{A^i v\}_{i=0}^{n-1}$ forms a
basis for $\mathbb{R}^{n}$. Then it follows that  
$Bv=\sum_{i=0}^{n-1} b_i A^i v$ for some scalars $b_i$.  Since $B$ commutes
 with $A$ by assumption it also commutes with all powers of $A$. Hence for all
 integers $r\geq 0$
 \begin{equation}
   \label{eq:18}
    B (A^r v) = A^r B v = A^r \sum_{i=0}^{n-1} b_i A^i v =
    \sum_{i=0}^{n-1} b_i A^i (A^r v)
 \end{equation} 
Since the set $\{A^rv\}_{i=0}^{n-1}$ span $\mathbb{R}^n$ we have shown that for any vector $x$ we
have $Bx = (\sum_{i=0}^{n-1} b_i A^i )x$ which imply (ii). 
\end{proof}

The result in Theorem~\ref{thm:commute} shows a direct limitation for MIMO
systems. The set of matrices that commutes with $A$ is only an $n$-dimensional
subspace of all $n\times n$ matrices. Hence we can only fix $n$ parameters in
$C_0$, e.g.\ one row as shown in the proof of Lemma~\ref{lem:1}. The remaining
rows in $C_0$ must be free parameters to optimize over. For MIMO systems it is
thus beneficial to iterate between minimizing $V$ w.r.t. $C,D$ and $B,D$
respectively until convergence.

% The case of SISO and MISO systems for time-domain estimation cost function
% Eq. \eqref{eq:3} is presented below. It is shown that the choice of fixed $C$ is
% a scaling on $B$ for these systems. This does not affect the cost function
% value, and for any $C$ that satisfy a mild condition the same cost function
% value is attained. The case for SIMO systems can be shown using the same
% argument, where the choice of $B$ is a scaling on $C$.

Theorem~\ref{thm:AC} and Theorem~\ref{thm:commute} solely rely on the
observability assumption on the $(A,C)$ pair. One extra simplifying assumption
can be made to illustrate the underlying structure of the least squares problem
Eq.~\eqref{eq:3}. Assume $A$ has distinct eigenvalues. Then $A$ admits the
decomposition $A=P \Lambda P^{-1}$ with a diagonal matrix $\Lambda$. Rewrite
Eq.~\eqref{eq:3c} with $B\in\mathbb{R}^{n \times n_u}$,
$C\in\mathbb{R}^{1\times n}$, $D\in\mathbb{R}^{1 \times n_u}$ exactly as:
\begin{align*}
\hat y(t|A,B,C,D)&=Du(t)+CP\sum_ {k=1}^t \Lambda^{t-k}P^{-1}Bu(k)
\end{align*}
Let $\bar{C}=CP\in\mathbb{C}^{1\times n}$,
$\bar{B}=P^{-1}B\in\mathbb{C}^{n\times n_u}$ and
$L_{t-k}=[\lambda^{t-k}_1 \lambda^{t-k}_2 \dots
\lambda^{t-k}_{n}]\in\mathbb{C}^{1\times n}$ where $\lambda_i$ for
$i=\{1,\dots,n\}$ are the diagonals of $\Lambda$. For real-valued $A$, $B$, and
$C$ matrices the columns of $\bar{C}$ and the rows of $\bar{B}$ come in
conjugate-pairs. Since $\Lambda^{t-k}$ is diagonal,
$\bar{C}\Lambda^{t-k}=L_{t-k}\;diag(\bar{C})$ where $diag(\bar{C})$ is the
diagonal matrix with elements of the vector $\bar{C}$ on the
diagonals. Therefore:
\begin{align*}
\hat y(t|A,B,C,D)&=Du(t)+\sum_ {k=1}^t L_{t-k}\;diag(\bar{C}) \bar{B}u(k)
\end{align*}
This form shows that the eigenvalues of $A$, contained in $L_{t-k}$, form a set
of basis functions for the linear regression problem for $B$ (equivalently,
$\bar{B}$). Fixed $C$, and in turn $diag(\bar{C})$, is a scaling of the basis
functions. Note that $\bar{C}\in\mathbb{C}^{1 \times n}$ having a zero element
is equivalent to $(A,C)$ pair not being observable when $A$ has distinct
eigenvalues. When an element of $\bar{C}$ is zero, the corresponding basis
function is multiplied by zero and not utilized in regression.

A less formal view of Theorem~\ref{thm:AC} and its proof can also be seen from
this form. Note that the unknowns in $B,C$ can be combined into a single matrix
$X=diag(\bar{C})\bar{B}$, and the linear regression can be performed for $X, D$
for fixed A, where the rows of $X$ are in conjugate-pairs if $A, B, C$ are
real-valued. $diag(\bar{C})$ is nonsingular when all elements of $\bar{C}$ are
nonzero.  Any fixed choice of $C$ that correspond to all nonzero elements in
$\bar{C}=CV$ (i.e. $(A,C)$ pair is observable when A has distinct eigenvalues)
can be used to extract a $B$ matrix estimate from $X$, without affecting the
cost function value observed in Eq.~\eqref{eq:3}.

The remarks made through distinct eigenvalues in $A$ assumption can be relaxed
by replacing the eigenvalue decomposition $A=P \Lambda P^{-1}$ with Jordan
matrix decomposition $A=S J S^{-1}$. $J$ is block diagonal, where the blocks
$J_i\in\mathbb{C}^{n_i \times n_i}$ are elementary Jordan blocks. Then
$A^{t-k}=SJ^{t-k}S^{-1}$ and $J^{t-k}$ is again block diagonal with the blocks
$J_i^{t-k}$ as upper triangular Toeplitz matrices. This structure of the
$J_i^{t-k}$ matrix blocks can be exploited to again collect the unknowns in
$B,C$ into a single matrix. Similarly, which of the basis functions are utilized
in the linear regression can be seen through (potentially a subset of) the
elements in $\bar{C}=CS$.

The optimality results for the $B, C, D$ matrices presented in this section are
also true for frequency-domain identification when a cost function of the
following form is utilized:

\begin{equation} \label{eq:freqDomainCost}
V(G,A,B,C,D)=\sum_k\|C(z_kI-A)^{-1}B+D-G_k\|_F^2
\end{equation}

The cost function Eq.~\label{eq:freqDomainCost} is typical for many estimation
approaches including the subspace methods. Here $z_k=e^{jT_s\omega_k}$ is a
point on the unit disk and $G_k$ is a frequency response measurement
corresponding to the $k^{th}$ frequency point $\omega_k$ $rad/s$. It follows
that since also this cost function only depends on the input-output properties
of the state-space system the same argumentation as made for the time-domain
case also holds here.
%%%%%The proof is omitted for space constraints.
% It follows
% the same approach of decomposing $A=V \Lambda V^{-1}$ and showing $C$ is a
% scaling on $B$ for SISO and MISO systems, or that $B$ is a scaling on $C$ for
% SIMO systems.  

\section{Discrete-Time Models with Time Domain Data}
\label{sec:td}
The following experiment was performed in \cite{matlab18a}:

\begin{enumerate}
\item Generate 200 discrete-time systems randomly with \texttt{drss}. The systems
  are of order 7 with 4 inputs and 4 outputs. The feedthrough matrix
  $D$ was set to zero.
\item Simulate each system with a Gaussian white input with unit variance
  matrix. No attempt was made to tailor the input to the system properties. The
  noise-free data are denoted by \texttt{z}. Add white Gaussian noise with unit
  covariance matrix to the output, giving the estimation data sequence
  \texttt{zn}.
\item For each dataset, estimate a subspace model of order 7 without disturbance
  model with \texttt{n4sid}, giving
  \texttt{mn=n4sid(zn,7,'disturbancemodel','none')}. Compute the error in that
  model compared to the true system by executing the commands:
  \begin{verbatim}
    en=pe(z,mn); 
    mnn=norm(en.y'*en.y/N)
  \end{verbatim}
  \vspace{-.15in}
\item Readjust the $B, C$ parameters to data:
  \begin{verbatim}
    mpBC=nn;
    mpBC.Structure.A.Free=zeros(7,7);
    mpBC=ssest(zn,mpBC);
  \end{verbatim}
\vspace{-.15in}
\item Compute the error of that model: 
  \begin{verbatim}
    en=pe(z,mpBC);
    mnp=norm(en.y'*en.y/N)
  \end{verbatim}
\vspace{-.15in}
\item Readjust the $A, B, C$ parameters of the model by further
  Maximum Likelihood iterations on \texttt{mn}: 
  \begin{verbatim}
    mp=ssest(zn,mn);
    en=pe(z,mp);
    mnp=norm(en.y'*en.y/N)
  \end{verbatim}
\vspace{-.15in}
\end{enumerate}

The results are summarized in the boxplot in Fig.~\ref{fig:box}. The
medians of the error norms are 

\texttt{mn} : 0.0370\\
\texttt{mpBC}:  0.0338\\
\texttt{mp}: 0.0267

The error for \texttt{mn} was larger than \texttt{mpBC} in 73.5 \% of
the cases. The error in \texttt{mnBC} was larger than \texttt{mp} in
80.5 \% of the cases and \texttt{mp} outperformed \texttt{mn} for 83 \%
of the data sets.

A scatter plot for the error norms for the 200 tested systems is given
in Fig.~\ref{fig:scatter}. It is in accordance with the observations
in \cite{Ljung:03ssp}.
\begin{figure}
  \centering
  \includegraphics[scale=0.4]{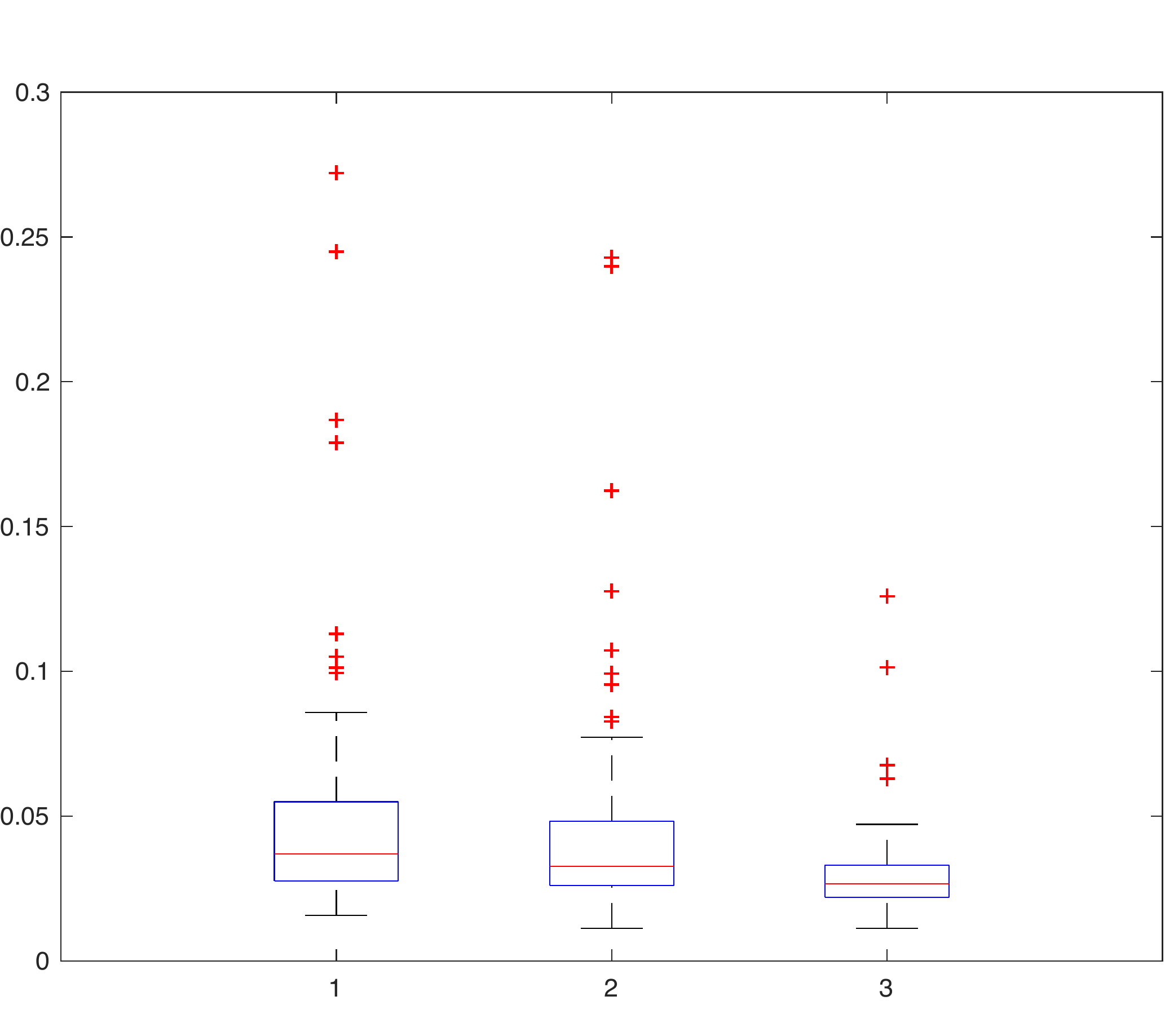}
  \caption{Boxplots for the error norms in Section \ref{sec:td}. From left to right:
    \texttt{n4sid}, \texttt{mpBC, mp}. The error norms were significantly above
    the plot limit for 9 systems for \texttt{n4sid} and \texttt{mpBC}.}
\label{fig:box}
\end{figure}
\begin{figure}
  \centering
  
\includegraphics[scale=0.4]{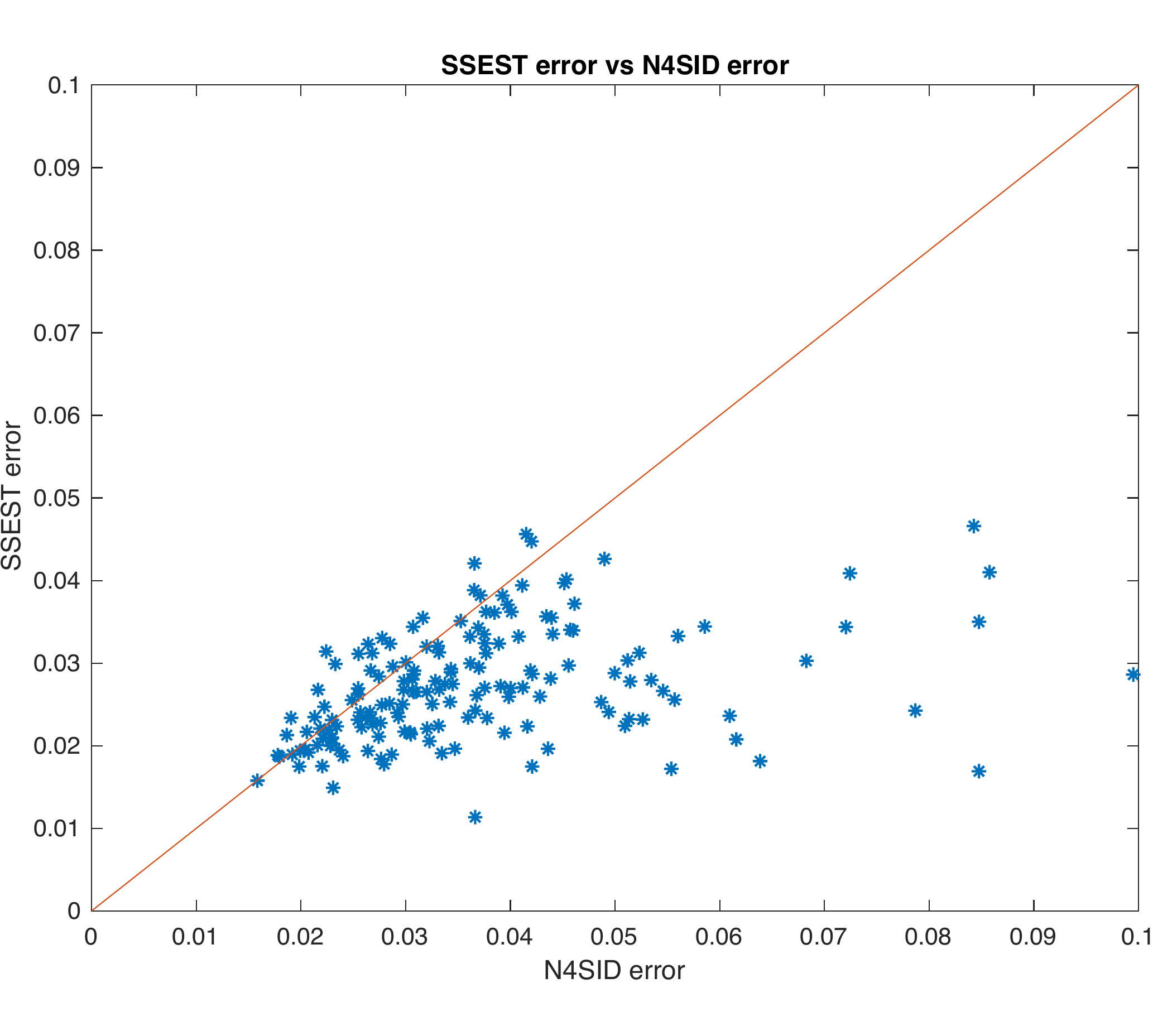}
  \caption{Plot of the error norm of the N4SID model (x-axis) vs the
    SSEST adjusted model (y-axis). 83\% of the 200 points are below
    the equal-line.}
\label{fig:scatter}
\end{figure}

\section{Continuous-Time Models with Frequency Domain Data} \label{sec:fd}
Similar experiments were carried out in the frequency domain with
continuous-time systems.
\begin{enumerate}
\item Generate 200 seventh order continuous-time systems with 4 input
  and 4 outputs using \texttt{G = rss(7,4,4)}   possibly with feedthrough terms $D$.
 Compute their frequency
  response functions \texttt{Gf = idfrd(G,FG)} at 410 linearly spaced
    frequencies \texttt{FG}. No attempts were made to adjust \texttt{FG}
    to the dynamics of \texttt{G}. Add 20 \% random, multiplicative
    noise to the response to form the frequency domain data
    \texttt{Gfn}.
\item Estimate  7th order  state-space models from the noisy frequency
  response functions, using \texttt{n4sid} and the $B, C$ adjusted
  models \texttt{mnBC} and $A, B, C$ adjusted models \texttt{mp} and
  their errors compared to the true system as in the previous section.
\end{enumerate}

The medians of the errors were found to be:

\texttt{mn}: 48.768\\
\texttt{mpBC}: 44.191\\
\texttt{mp}: 31.485

The results are summarised in the boxplots in Fig.~\ref{fig:boxFD}.

The error for \texttt{mn} was larger than \texttt{mpBC} in 83.5 \% of the
cases. The error in \texttt{mnBC} was larger than \texttt{mp} in 74.5 \% of the
cases and \texttt{mp} outperformed \texttt{mn} for 87 \% of the data sets.

\begin{figure}
  \centering
  \includegraphics[scale=0.4]{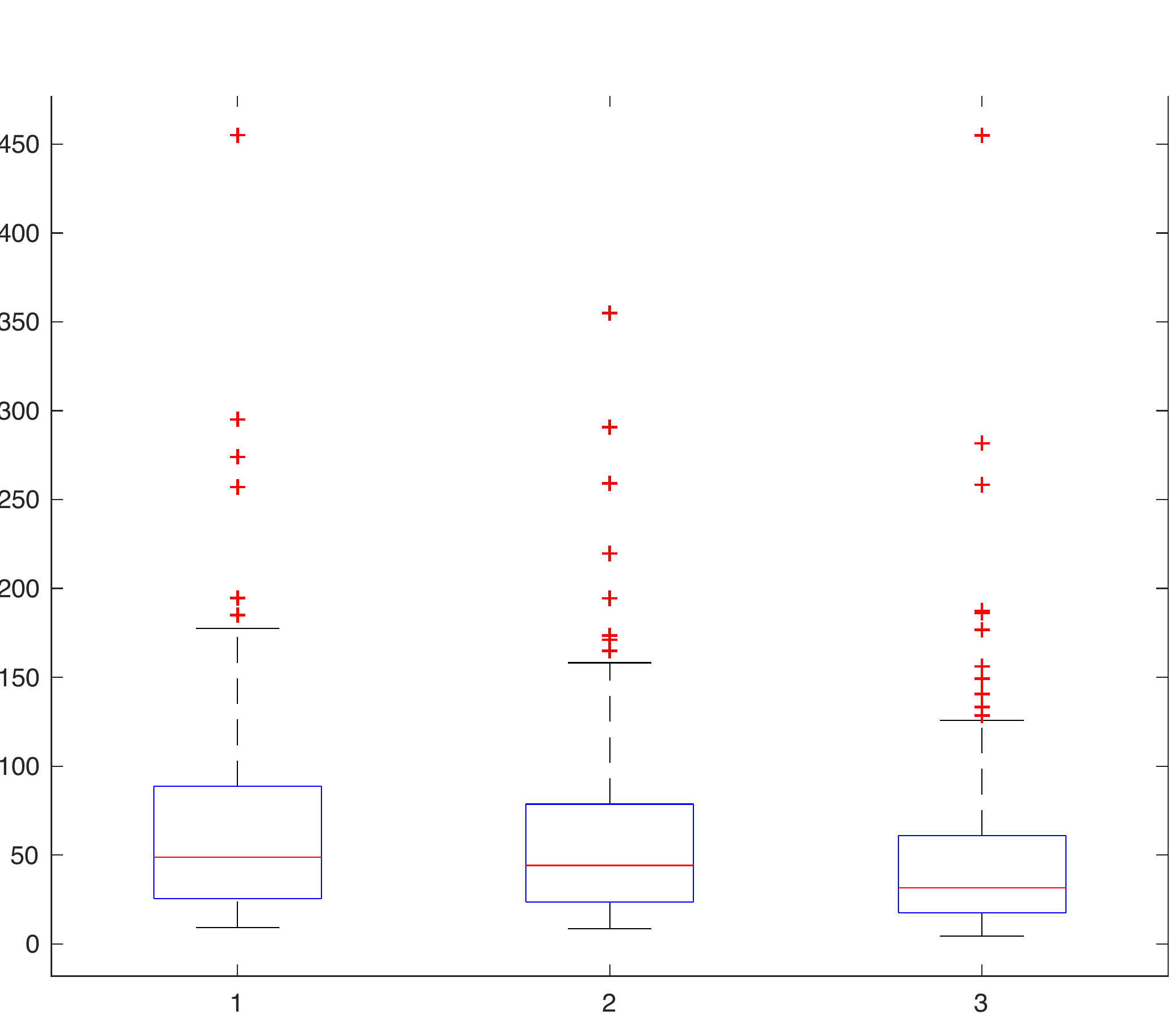}
  \caption{Boxplots fot the error norms in Section \ref{sec:fd}. From left to right:
    \texttt{n4sid}, \texttt{mpBC, mp}. }
  \label{fig:boxFD}
\end{figure}
A scatter plot for the error norms for the 200 tested systems is given in
Fig.~\ref{fig:FDscatter}.
\begin{figure}
  \centering
  
\includegraphics[scale=0.4]{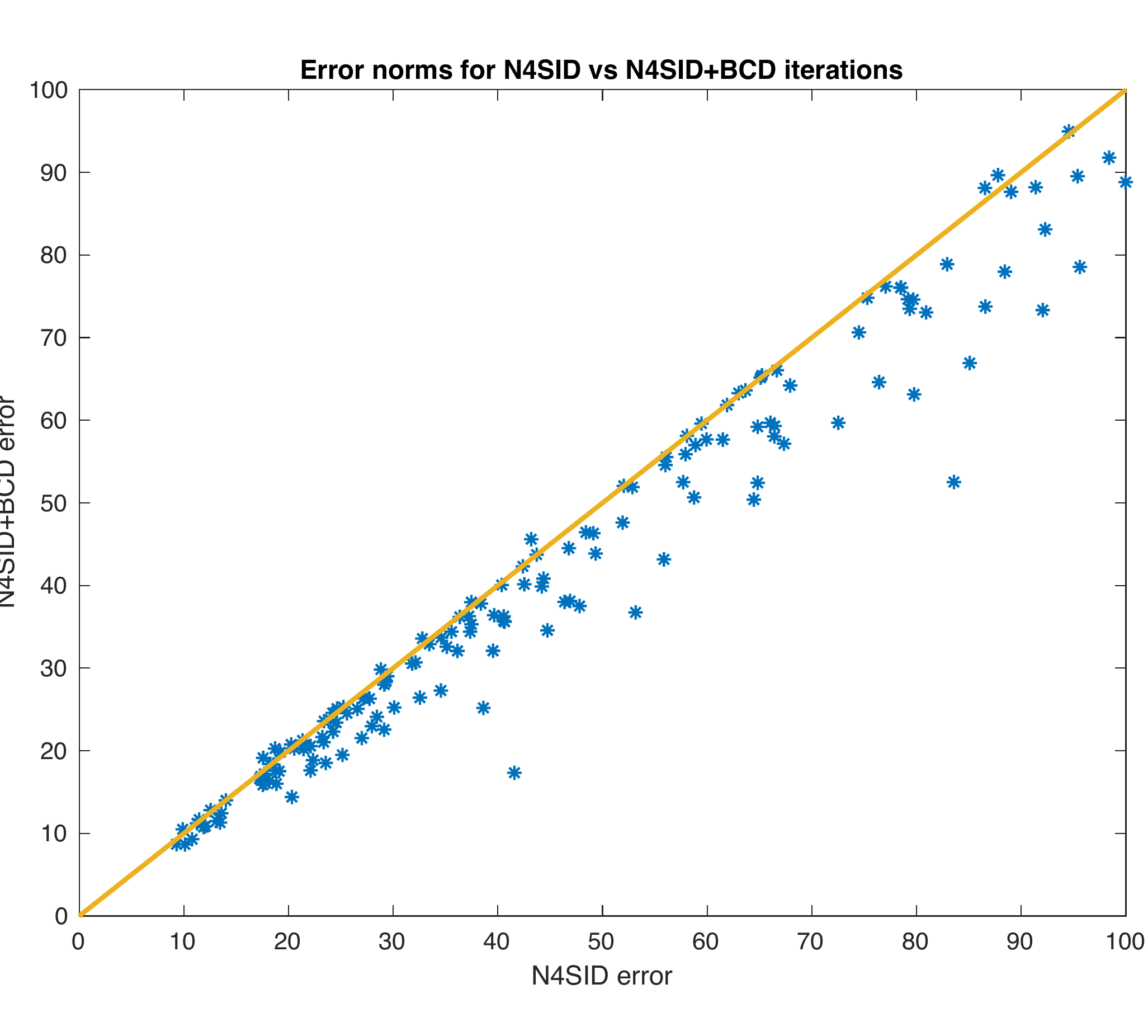}
  \caption{Plot of the error norm of the N4SID model (x-axis) vs. the
    B,C,D adjusted N4SID model (y-axis). 83.5\% of the 200 points are below
    the equal-line.}
\label{fig:FDscatter}
\end{figure}

\section{Real Data from Flexible Structures} \label{sec:xc90} Mechanical
vibration testing was performed on a Volvo XC90 (2015) rear subframe
structure. During the testing the subframe was equipped with 26 uniaxial and 10
triaxial accelerometers yielding a total of 56 measurement channels. A shaker
was used to provide excitation at two different locations. Frequency data was
obtained by employing a single input multiple output (SIMO) stepped sine testing
procedure for each shaker location. This testing directly yields measurements of
the frequency response function at the excited frequencies.  For each of the
shaker locations, a total of 2998 different frequencies were excited in the
frequency range between 60 and 500 Hz with a frequency spacing derived according
to the method described in \cite{vakilzadeh2015experiment}. For further details
of the experiment, refer to \cite{gibanica_parameter_2017}. Using the frequency
domain subspace method described in \cite{McKelveyAkcayLjung:96}, an initial
continuous-time model was derived of order 40 with 56~outputs and 2~inputs.
Subsequently the $C$ and $D$ matrices are iteratively reestimated followed by an
estimate of the $B$ and $D$ matrices. After the first $B, C, D$ iteration the
least-squares loss function is reduced to 58.5\% of the initial value. After
four additional $B, C, D$ iterations the loss function is slightly decreased to
56.9\%.  In Fig.~\ref{fig:xc90} the magnitude of the final frequency response
function corresponding to input 1 and output 7 is plotted together with the
error magnitude, that is, the magnitude of the difference between the data and
model. The error magnitude of the initial model is also plotted in the
graph. For this rather large-scale example it is clear that it is advisable to
employ at least a first $B, C, D$ iteration after the initial estimate delivered
by the subspace method.

% iteration: 1 Error  % of inital: 0.58482
% iteration: 2 Error  % of inital: 0.57223
% iteration: 3 Error  % of inital: 0.57005
% iteration: 4 Error  % of inital: 0.56946
% iteration: 5 Error  % of inital: 0.56922
 \begin{figure}
   \centering
  \includegraphics[width=\columnwidth]{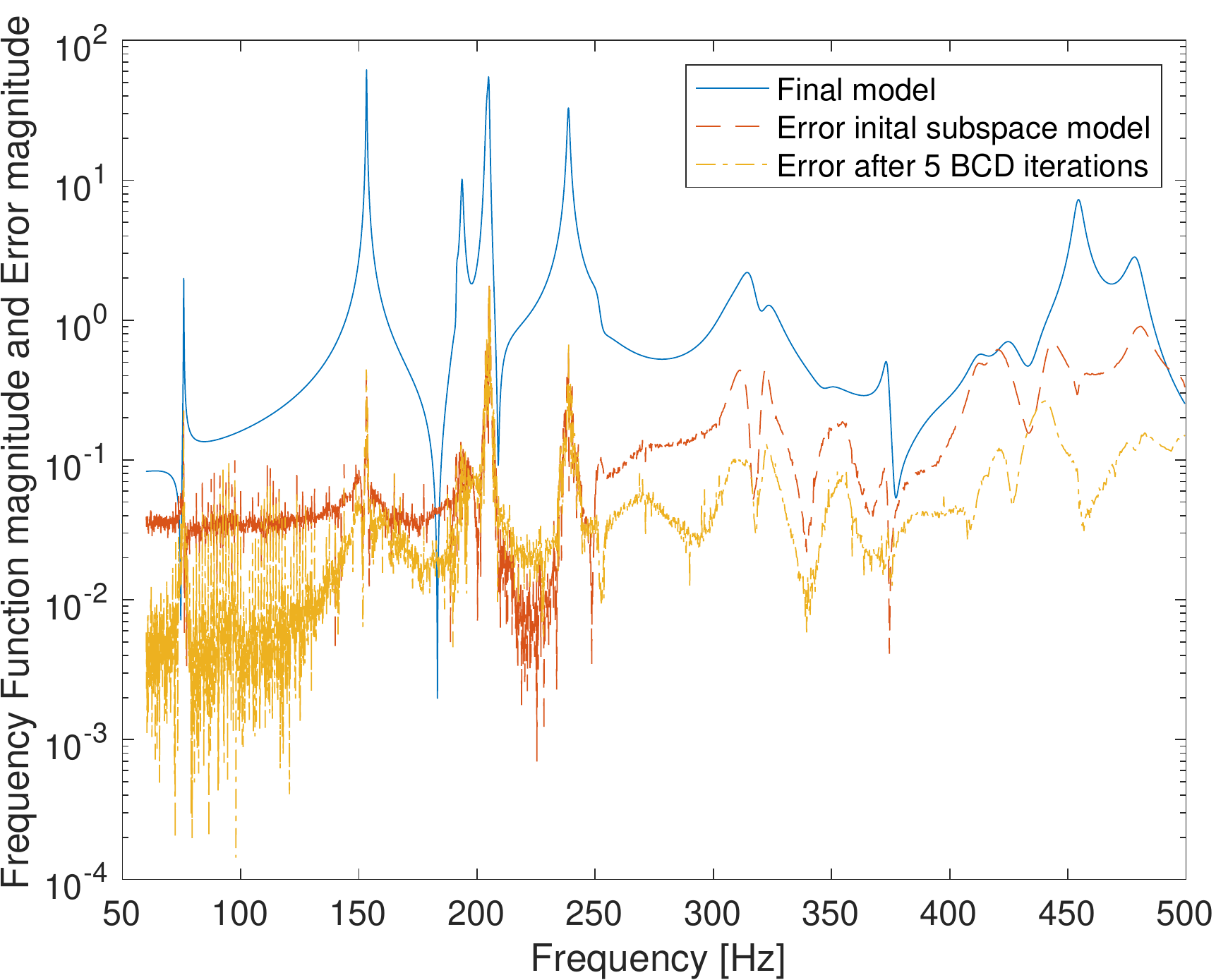}
   \caption{Magnitude of frequency response function and error
     magnitude of estimated model of order 40.}
   \label{fig:xc90}
 \end{figure}

Fig.~\ref{fig:xc90_2} compares the evolution of the cost function in
Eq.~\eqref{eq:freqDomainCost} over five steps for: (i) $B,C,D$ iterations with
fixed A, (ii) Joint $B,C,D$ minimization via Gauss-Newton method with fixed A,
(iii) Joint $A,B,C,D$ minimization via Gauss-Newton method. A subset of the
input channels, the first four of the 56 total, were used for computation
speed. A $40^{th}$ order model estimate from \texttt{n4sid},
\cite{McKelveyAkcayLjung:96}, is used as the initial model for all
approaches. The cost function value obtained in the subsequent steps for all
methods are normalized by the value attained by this initial step. For this
specific data and the initial model delivered by the subspace methods, the 
\textit{block coordinate descent} approach $B, C, D$ iterations approximately
converge in one step to a (local) minimum. This method was also the
fastest. Gauss-Netwon approach for optimizing $B, C, D$ matrices jointly
approximately converged in two steps. Optimizing over $A, B, C, D$ jointly can
potentially find a better (local) minimum compared to the previous two methods
since the $A$ matrix is also optimized, but the progress was slower. Indeed once
any optimization approach for $B, C, D$ iterations converges to a minimum, this could be followed with an
optimization over the $A$ matrix. 

\begin{figure}
  \centering
 \includegraphics[width=\columnwidth]{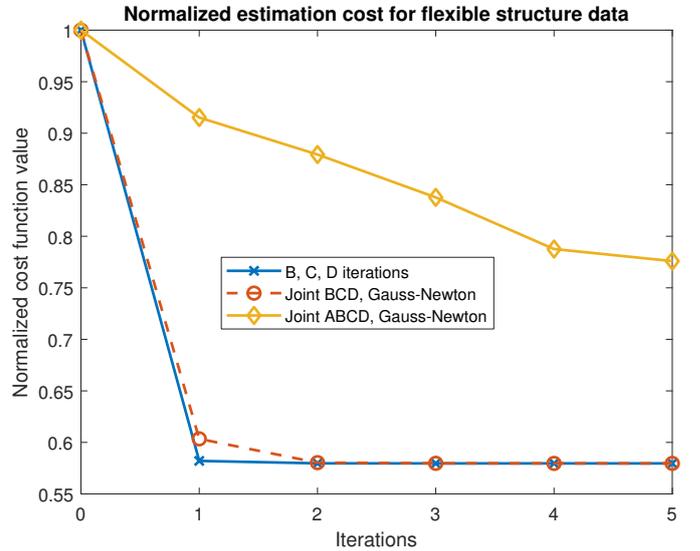}
 \caption{Evolution of the cost function values for three optimization
   approaches. The values are normalized by the value attained by the initial
   model obtained from the subspace methods.}
  \label{fig:xc90_2}
\end{figure}

\section{Conclusions}
This contribution shows that it is useful to do further iterations on the $B$,
$C$, and $D$ state-space matrix estimates obtained by the subspace
identification when the system has multiple outputs. This is tested with both
simulated and real data. It is also recommended to adjust the $A$ estimate using
maximum likelihood iterations. It leads to improvements in a majority of cases,
but not always. One reason is that subspace identification algoritms contain
several design variables (prediction horizons and prefilters) that can be
difficult to choose in an optimal way. In a sense, maximum likelihood iterations
(such as in \texttt{ssest}) perform such choices automatically.

\paragraph*{Acknowledgement} The authors would like to thank Thomas Abrahamsson
at the Department of Mechanics and Maritime Sciences, Chalmers University of
Technology for sharing the experimental data.
%\bibliography{../../../bib/ref2M.bib}
\bibliography{./ref2M}

\end{document}